\documentclass[12pt,a4paper]{article}
\usepackage{hyperref}
\usepackage[
  margin=2.6cm,
  includefoot,
  bmargin=1.3cm,
  footskip=1.3cm
]{geometry}
\usepackage{float}
\immediate\write18{./utils/sync_bibliography.sh}
\usepackage[T1]{fontenc}
\usepackage{lmodern}
\newcommand{\mathe}{\ensuremath{\mathrm{e}}}
\newcommand{\mathi}{\ensuremath{\mathrm{i}}}
\usepackage{textgreek}
\usepackage{mathtools}

\newcommand{\transpose}[1]{\ensuremath{#1}^{\mathsf{T}}}
\newcommand{\supth}{\textsuperscript{th}\ }
\usepackage{amssymb} 
\usepackage{amsmath,mathtools} 
\usepackage{amsthm}
\theoremstyle{plain}
\newtheoremstyle{slanted}
  {}
  {}
  {\slshape}
  {}
  {\bfseries}
  {}
  { }
  {}
\theoremstyle{slanted}
\newtheorem{thm}{Theorem}
\newtheorem{lem}{Lemma}
\newtheorem{prop}{Proposition}

\newtheorem{cor}{Corollary}
\newtheorem{claim}{Claim}
\theoremstyle{definition}
\newtheorem{defn}{Definition}
\newtheorem{obs}{Observation}

\theoremstyle{remark}

\usepackage{xcolor}%
\hypersetup{
    colorlinks=true,       
    linkcolor=cyan!40!black,          
    citecolor=cyan!40!black,        
    filecolor=cyan!40!black,      
    urlcolor=cyan!40!black,           
    runcolor=cyan!40!black
}
\usepackage{verbatim}
\usepackage{graphicx}
\usepackage[font={small}]{caption}
\usepackage{dcolumn}
\usepackage{stackengine}
\usepackage{setspace}
\usepackage{algpseudocode,algorithm,algorithmicx}
\floatstyle{plain}
\newfloat{myalgo}{tbhp}{mya}
{\begin{myalgo}[#1]
\centering
\begin{minipage}{#2}
\begin{algorithm}[H]\setstretch{1.15}}%
{\end{algorithm}
\end{minipage}
\vspace{0.5cm}
\end{myalgo}}
\usepackage{enumerate}
\usepackage{enumitem}
\usepackage{hhline}
\usepackage{array}
\newcolumntype{P}[1]{>{\centering\arraybackslash}p{#1}}
\usepackage{fancyhdr}
\usepackage{chngcntr}
\usepackage{tikz-cd}
\usepackage{physics}
\usepackage{complexity}
\usepackage{optidef}
\usepackage{dsfont}

\newcommand{\Prob}[1]{\ensuremath\operatorname{Pr}\left( #1\right)}

\newcommand{\relmiddle}[1]{\mathrel{}\middle#1\mathrel{}}

\newcommand{\clift}[1]{\ensuremath{#1}^{c}}
\newcommand{\dlift}[1]{\ensuremath{#1}^{d}}

\newcommand{\arange}[1]{\left\{0,\ldots,{#1}-1\right\}} 

\begin{document}
\title{For every quantum walk there is a (classical) lifted Markov chain with faster mixing time}
\author{Danial Dervovic\thanks{Department of Computer Science, University College London, Gower Street, London WC1E 6BT, United Kingdom; mailto:  \href{mailto:d.dervovic@cs.ucl.ac.uk}{\texttt{d.dervovic@cs.ucl.ac.uk}}.}}
\date{\today}
\maketitle

\begin{abstract}
    Quantum walks on graphs have been shown in certain cases to mix quadratically faster than their classical counterparts. Lifted Markov chains, consisting of a Markov chain on an extended state space which is projected back down to the original state space, also show considerable speedups in mixing time. Here, we construct a lifted Markov chain on a graph with $n^2 D(G)$ vertices that mixes exactly to the average mixing distribution of a quantum walk on the graph $G$ with $n$ vertices, where $D(G)$ is the diameter of $G$. Moreover, the mixing time of this chain is $D(G)$ timesteps, and we prove that computing the transition probabilities for the lifted chain takes time polynomial in $n$. As an immediate consequence, for every quantum walk there is a lifted Markov chain with a faster mixing time that is polynomial-time computable, as the quantum mixing time is trivially lower bounded by the graph diameter.
    The result is based on a lifting presented by Apers, Ticozzi and Sarlette (\texttt{arXiv:1705.08253}). 
\end{abstract}

\section{Introduction}

Sampling from a desired probability distribution over a given state space is an important computational task, used in many diverse fields.
Markov chain methods have proven to be widely successful in this domain being used for applications such as approximating the permanent of a matrix~\cite{Sinclair1993}, machine learning~\cite{Andrieu2003} and analysing the performance of distributed systems~\cite{Molloy1982}.
In practise, many approaches suffer from a lack of provable upper bounds on the time it takes to draw samples.
One such class of methods is Markov Chain Monte Carlo (MCMC), where one conducts a specific random walk on the state space of interest for some set number of timesteps, then measures the position of the walker.
Often the user of an algorithm in the MCMC framework is unsure if the chain has \emph{mixed}, that is, is sampling the walker's position equivalent to sampling the desired distribution?
More precisely, is the distribution over the vertices close in total variation distance to the stationary distribution of the Markov chain?
In most cases, the answer to this question is unknown, the practitioner empirically determines a favourable time to run the chain for, without any theoretical guarantee of closeness to the desired distribution~\cite{Koller2009}.

One goal of Markov chain theory is to provide concrete upper bounds for mixing times of the Markov chains used in MCMC and more generally, mixing times for an arbitrary Markov chain.
In this paper we restrict the discussion to discrete state spaces, in which case a Markov chain is most naturally described as a random walk on a graph.
There is a proven lower bound on the mixing time of a Markov chain on any graph, $\Omega(1/ \Phi)$, where $\Phi$ is the conductance of the graph in question~\cite{Sinclair1993}. 
Two techniques that have been introduced in an effort to speed up mixing of Markov chains are quantum walks~\cite{Aharonov2001,Richter2007} and lifted Markov chains~\cite{Chen1999,Diaconis2000}.
With quantum walks, quantum superposition is employed to decrease mixing time, but requires use of a quantum computer to work in practise.
In lifted Markov chains, one carries out a random walk on a graph homomorphic to the original and projects down to the original graph at the end of the walk.

In this work, we construct a lifted Markov chain that mixes to the probability distribution induced by a quantum walk, Ces\`{a}ro averaged over $T \to \infty$ timesteps. 
We prove that the lifted chain mixes exactly to this distribution in time equal to the diameter of the graph upon which the quantum walk takes place.
Moreover, we show that computing the lifting takes time $O(n^8)$, where $n$ is the number of vertices in the graph.
Intuitively, this means that a lifted chain can be constructed that simulates the mixing of a quantum walk in a shorter time it takes to carry out the walk.
However, using this lifting only confers an advantage over the native quantum walk if the quantum walk takes $T=\Omega(n^8)$ timesteps, taking into account computation of the transition probabilities.
More precisely, our lifted chain mixes to the \emph{average mixing distribution} of a quantum walk of choice; the full result is given in Theorem~\ref{thm:main_result} and Corollary~\ref{cor:main_result}.
The average mixing distribution after $T$ timesteps corresponds to sampling uniformly at random a time $t \in \{0,1,\ldots, T-1\}$, running the quantum walk for $t$ timesteps, then measuring the position of the walker.
This procedure is employed instead of simply running for $T$ timesteps then measuring, as the latter process doesn't converge in the limit of infinite $T$.
This is discussed in further detail in Section~\ref{subsec:mixingtimes}.

The proof of Theorem~\ref{thm:main_result} proceeds in the following manner: we begin with a quantum walk on the graph $G$ over $T$ timesteps.
We then use a lifting defined by Apers, Ticozzi and Sarlette in~\cite{AST17-l} which we call the $d$-lifting, that allows diameter-time mixing to any distribution with full support, taking the quantum average mixing distribution as the target distribution.
We further prove that the runtime of computing this lifting is polynomial in $n$.

The layout of the paper is as follows: in Section~\ref{sec:premiliminaries} we define Markov chains, quantum walks, lifted Markov chains and make precise the notion of mixing.
In Section~\ref{sec:cycle_example} we discuss mixing on an example graph, the cycle, for clarity.
In Section~\ref{sec:equivalence} we introduce and prove the necessary ingredients for Theorem~\ref{thm:main_result}.
Section~\ref{sec:discussion} concludes the paper with discussion and open questions.

\subsection{Related Work}

Upon completion of the first version of this work (\href{https://arxiv.org/pdf/1712.02318v1.pdf}{\texttt{arXiv:1712.02318v1}}), the author became aware of the extended abstract by Apers, Sarlette and Ticozzi in~\cite{AST17-a}, which presents a similar result to the one proved here.
Their result stated that for any local-stochastic process (a quantum walk is local-stochastic) that mixes to a distribution $\pi$ in time $\mathcal{M}_{\overline{\epsilon}}$ (our notation), there is a lifted chain that mixes to $\pi$ in time $\mathcal{M}_{\overline{\epsilon}}$ with exponential convergence to arbitrary total variation distance  $\epsilon > 0$ from $\pi$.
The proof was not publicly available at that time.
They have subsequently released the paper~\cite{AST17-q} proving the result.
In the paper~\cite{AST17-q} there is no discussion of the computational complexity of their constructions; we use some of their results to show that computing the lifting presented in this paper is polynomial-time.

\section{Preliminaries}\label{sec:premiliminaries}

We denote the set of nonnegative integers by $\mathbb{Z}_+$.
For a natural number $n \in \mathbb{N}$, $[n]:=\{1,2,\dots, n\}$.
We denote by $\pm 1$ the set $\{1, -1\}$.
The function $\delta^i_j$ is the Kronecker delta, whence $\delta^i_j =1$ if and only if $i=j$.
The matrix $I_n$ is the $n \times n$ identity matrix and $\vb*{1}_n$ is the all-ones (column) vector with $n$ elements, we omit the subscript if the dimensionality is clear from context.
For a set $S$, we will write $2^S$ to denote the set of all subsets of $S$.

\subsection{Markov Chains}

Consider a directed graph $G=(V(G),E(G))$ on $n$ vertices with vertex set $V(G)=[n]$ and arc set $E(G)\subseteq V(G) \times V(G)$.
We call $G$ a \emph{symmetric} directed graph if $(x,y) \in E(G) \Leftrightarrow (y,x) \in E(G)$ for all $x,y \in V(G)$ and say that $x$ and $y$ are \emph{adjacent}.
We denote by $x \sim y$ the adjacency of two vertices $x$, $y$.
A graph is called \emph{$m$-regular} if every vertex has $m$ neighbours, where $m\in\arange{n}$.

We can define a \emph{discrete-time Markov chain} $M_G$ on the vertices of $G$ as follows:
Let $X(t)$ be a random variable, where $X(t) \in V(G)$, for all $t\in \mathbb{Z}_+$.
The Markov chain $M_G$ is the sequence of states $(X(0), X(1), \ldots)$, that additionally satisfies the following properties:
\emph{i.} The starting state $X(0)$ is distributed according to an initial probability distribution $p^{(0)}$.
\emph{ii.} The probability of observing state $X(t)$ is independent of all previous states, apart from its immediate predecessor, $X(t-1)$.
\emph{iii.} The \emph{transition probability} between states $i,j \in V(G)$, $\Prob{ X(t+1) = i \mid X(t) = j } > 0$ if and only if $(j,i)\in E(G)$. 

From the above definition, each arc $(i,j)$ in $G$ has an associated transition probability from vertex $i$ to vertex $j$. 
The arc probabilities must be nonnegative and the sum of the probabilities leaving a given vertex must equal one.
These transition probabilities are listed in the matrix $P\in \mathbb{R}^{n\times n}$, where
\begin{equation}
    P_{ij} = \Prob{ X(t+1) = i \mid X(t) = j },\ i,j\in [n].
\end{equation}
This matrix must satisfy
\begin{equation}\label{eq:mrkvchaincond}
    P_{ij} \geq 0, \qq{} \transpose{\vb*{1}} P = \transpose{\vb*{1}}, \qq{} P_{ij} = 0 \Leftrightarrow (j,i)\notin E(G).
\end{equation}
The conditions in Eq.~\eqref{eq:mrkvchaincond} state that $P$ must be a column-stochastic matrix with support only on elements corresponding to arcs in the the graph $G$.
From the above, at a time $t$, the distribution over states will be
\begin{equation}\label{eq:ClMrkvEvolv}
    p^{(t)} = P^t p^{(0)},
\end{equation}
where a probability distribution $p$ is represented as a column vector.
In other words, $X(t)$ is a random variable distributed according to $p^{(t)}= P^t p^{(0)}$.

We see that a Markov chain $M_G$ on a finite, directed graph $G$ can be completely characterised by the transition matrix $P$ and the initial state distribution $p^{(0)}$, so we shall use the shorthand $M_G = (P, p^{(0)})$.
Note that we will use the terms `Markov chain' and `random walk' interchangeably.
Often a Markov chain starts at a particular vertex $j$, in which case $p^{(0)} = \vb{e}_j$, where $\vb{e}_j$ is the $j$\supth standard basis vector.
We also note that often in the literature, probability vectors $p^{(t)}$ are row vectors and transition matrices act from the right, in contrast to our definitions.

\subsection{Lifted Markov Chains}\label{sec:liftedMarkovchains}

Lifted Markov chains were first introduced by Chen, Lov\'{a}sz and Pak in \cite{Chen1999} as a mechanism to reduce the mixing time of a Markov chain on a given graph $G$.
A graph $\clift{G}$ is a \emph{lift} of $G$ if there exists a homomorphism $c: \clift{G} \to G$.
Following Apers, Ticozzi and Sarlette~\cite{AST17-l}, we denote by $c^{-1}: V(G)\to 2^{V(\clift{G})} $ the map that takes as input the vertex $k \in V(G)$ and outputs the set of nodes $j\in V(\clift{G})$ for which $c(j)=k$.
The homomorphism $c$ induces a linear map from $V(\clift{G})$ into $V(G)$, which we can represent using the matrix $C$ with elements
\begin{equation}
    C_{i,j} = 
    \begin{cases}
        1, & \text{if } c(j) = i; \\
        0, & \text{otherwise},
    \end{cases}
\end{equation}
where $i\in V(G)$, $j \in V(\clift{G})$.
We can now define a lifted Markov chain.

\begin{defn}
    \emph{($c$-lifted Markov chain)}
    Let $G$ be a finite, directed graph and let $M_G=(P, p^{(0)})$ be a Markov chain on $G$.
    Furthermore, th graph $\clift{G}$ is a lift of $G$ via a homomorphism $c: \clift{G} \to G$.
    A \emph{$c$-lifted Markov chain} for $M_G$, $\clift{M_G}$, is the Markov chain $(\clift{P}, \clift{p}{}^{(0)})$ on the graph $\clift{G}$ that satisfies the following properties:
    \begin{enumerate}
        \item The transition matrix $\clift{P}$ satisfies $P C = C \clift{P}$.
        \item The initial distribution $\clift{p}{}^{(0)}$ satisfies $p^{(0)} = C\cdot  \clift{p}{}^{(0)}$.
    \end{enumerate} 
\end{defn}
In the notation of category theory, one can say that that the following diagram commutes:
\begin{equation*}
\begin{tikzcd}
V(G) \arrow{r}{P}  & V(G) \\%
V(G^c) \arrow{r}{P^c} \arrow{u}{C} & V(G^c) \arrow{u}{C} 
\end{tikzcd}
\ \  .
\end{equation*}
The lifted Markov chain $\clift{M}_G$ proceeds in the usual way, by repeated application of $\clift{P}$.
The probability distribution over $V(G)$ is given at time $t$ by the marginal $p^{(t)} = C (\clift{p})^{(t)}$.
We shall call $M_G$ the \emph{coarse-grained chain} with respect to the lifted chain $\clift{M_G}$. 

The definition of a $c$-lifting gives some freedom for the form of $\clift{P}$ and $\clift{p}{}^{(0)}$, even for a fixed homomorphism $c$.
Usually, we will specify the graph $\clift{G}$, transition matrix $\clift{P}$ and initial distribution $\clift{p}{}^{(0)}$ and refer to this specific configuration as \emph{the} $c$-lifting. 

\subsection{Coined Quantum Walks}

Suppose we have a $m$-regular graph $G$.
Define a Hilbert space associated to the vertices of $G$, $\mathcal{H}_{V(G)} = \operatorname{span}(\{\ket{v}\}_{v\in V(G)})$.
Also define a Hilbert space associated to the coin $\mathcal{H}_{C} = \operatorname{span}(\{\ket{k}\}_{k=1}^m )$.
Our quantum walk acts on the Hilbert space $\mathcal{H}_C \otimes \mathcal{H}_{V(G)}$.

We need two unitary operators to define a coined quantum walk, the \emph{coin operator} and the \emph{shift operator}.
We introduce the coin first: the coin $C$ is a unitary operator on $\mathcal{H}_C$.
A common coin operator is the \emph{Hadamard coin}, $H_m$, given by
\begin{equation}
    H_m = \frac{1}{\sqrt{m}} \sum_{j \in [m]} \sum_{ k\in [m]} \omega^{(j-1)(k-1)} \ketbra{j}{k},
\end{equation}
where $\omega:= \mathe^{\frac{2\pi \mathi}{m}}$.
We call a coined quantum walk utilising the Hadamard coin a \emph{Hadamard walk}.

We need one more piece to define a coined quantum walk, the shift operator $S$, for which we use the description of Godsil~\cite{Godsil2017}.
First, for each vertex $u$ we must specify a linear order on its neighbours
\begin{equation}
    f_u: \{ 1,2,\ldots, m \} \to \{ v: (u,v)\in E(G) \}.
\end{equation}
The vertex $f_u(j)$ will be referred to as the $j$\supth neighbour of $u$ and the arc $(u, f_u(j))$ the $j$\supth arc of $u$.
For each vertex $u$, the shift operator $S$ maps its $j$\supth arc to the $j$\supth neighbour of $f_u(j)$, i.e.\  $S(\ket{j} \otimes \ket{u}) = \ket{j} \otimes \ket{f_u(j)}$.

We can now construct one step of a coined quantum walk, described by the unitary operator $U = S\cdot (C \otimes I_{\mathcal{H}_{V(G)}})$.
An initial state of the walk is some unit vector $\ket{\psi(0) } \in \mathcal{H}_C \otimes \mathcal{H}_{V(G)}$, typically a basis state $\ket{k , v}$ for some $k\in [m]$, $v \in V(G)$, where we abbreviate $\ket{k}\otimes\ket{v}$ as $\ket{k,v}$.
The state after $t$ timesteps is $\ket{\psi_t} = U^t \ket{\psi(0) }$.
Thus we can totally characterise the state of a quantum walk by the tuple $(U, \ket{\psi(0) })$.

Following Aharonov \emph{et. al.}~\cite{Aharonov2001}\footnote{They use $P_t(v | \psi(0) )$, which we change to avoid notational clashes.}, we denote by $Q_t(v | \psi(0) )$ the probability of measuring the vertex $v$ at time $t$ of the quantum walk, contingent on the initial state being $\ket{\psi(0)}$.
More concretely,
\begin{equation}\label{eq:node_prob}
    Q_t(v | \psi(0) ) = \sum_{k \in [m]} \abs{ \bra{k , v} U^t\ket{\psi(0) }}^2.
\end{equation} 
We denote by $Q_t(\,\cdot\, | \psi(0) )$ the induced probability distribution over the vertices. 

In fact, we can define a \emph{general quantum walk} also, as in~\cite{Aharonov2001}.
In this case we relax the requirement of the exact form that $U$ can take, merely that $U$ must respect the structure of the graph.
More precisely, for any $k, v$, the quantity $U \ket{k,v}$ only contains basis states $\ket{k', v'}$ with $v' \in N(v) \cup \{v\}$, where $N(v)$ is the \emph{neighbourhood} of $v$, $N(v) = \{ u \,|\, (v,u) \in E(G) \}$.
The results we prove later hold for this general class of quantum walk.

\subsection{Mixing Times}\label{subsec:mixingtimes}

Suppose we have a Markov chain $M_G=(P, p^{(0)})$ over a finite directed graph $G$ and a probability distribution on the states of $M_G$, $\pi$ such that $P\pi = \pi$.
Then we call $\pi$ a \emph{stationary distribution} of $M_G$.
Indeed, $\pi$ exists and is unique if $M_G$ is \emph{irreducible} and \emph{aperiodic}~\cite{Levin2009}.
Moreover, an irreducible, aperiodic Markov chain always converges to the stationary distribution, that is, $\lim_{t\to\infty} P^t p^{(0)} = \pi$; a result known as the \emph{convergence theorem} in the literature.
Moreover, all of the elements of $\pi$ in this case are strictly positive.
Irreducibility of $M_G$ is equivalent to saying that the graph $G$ is connected.
The chain $M_G$ is aperiodic if there exists some time $T_0$ such that for all $t \geq T_0$ and all vertices $i,j \in V(G)$, $[P^t]_{i,j} > 0$.
A Markov chain that is irreducible and aperiodic is called \emph{ergodic}.

We now define the \emph{mixing time}, $\mathcal{M}_{\epsilon}$, of $M_G$, for $\epsilon >0$, 
\begin{equation}
    \mathcal{M}_{\epsilon} = \max_{p^{(0)} \in \mathcal{P}} \min_{T \in \mathbb{Z}_+} \qty{ T \relmiddle| \forall t \geq T,\ \norm{P^t p^{(0)} - \pi}_{TV} \leq \epsilon },
\end{equation}
where $\norm{\,\cdot\,}_{TV}$ is the total variation distance between two distributions $\pi$ and $\kappa$, that is, $\norm{\pi - \kappa}_{TV} = \frac{1}{2}\sum_i |\pi_i - \kappa_i| $ and $\mathcal{P}$ is the allowed domain of initial starting states.
Intuitively, the mixing time is the number of steps it takes for an arbitrary starting state to be $\epsilon$-close in total variation distance to the stationary distribution in the worst case.
Typically, $\mathcal{P}$ is the set of distributions with all probability mass on one and only one state, i.e.\  $\mathcal{P}=\qty{\vb{e}_i \relmiddle| i\in V(G) }$.
We can do this without loss of generality, since it can be shown that $\max_{x\in V(G)} \norm{P^t \vb{e}_x - \pi}_{TV} = \sup_{\mu} \norm{P^t \mu - \pi}_{TV}$, where $\mu$ is any probability distribution on $V(G)$~\cite[Exercise 4.1]{Levin2009}.

We say that the chain $M_G$ has \emph{mixed} at a time $T$ if  $\norm{P^T p^{(0)} - \pi}_{TV} \leq \epsilon$; from submultiplicativity of the $\ell_1$-norm the chain will be mixed for all $t\geq T$.
By convention, we shall often take $\epsilon = 1/4$.
Indeed, from \cite[Eq. (4.36)]{Levin2009} we can easily get a bound for arbitrary $0 < \epsilon < 1/4$, 
\begin{equation}
    \mathcal{M}_\epsilon \leq \left\lceil \log_2 \frac{1}{\epsilon} \right\rceil \mathcal{M}_{1/4}.
\end{equation}

The mixing time is strongly related to a topological property of the Markov chain called the \emph{conductance}.
We must first define the conductance of a Markov chain $(P,p^{(0)})$ on $G$. For a subset $X\subseteq V(G)$ let $\pi(X) = \sum_{i\in X}\pi_i$, where $\pi$ is the stationary distribution under $P$. The conductance $\Phi(P)$ of $P$ is defined as
\begin{equation}\label{eq:conductance}
    \Phi(P) = \min_{X\subset V(G);\, \pi(X)\leq \frac{1}{2} } \frac{\sum_{i\in X,\, j\not\in X} P_{j,i} \pi_i }{\pi(X)}.
\end{equation}
Often, the numerator of Eq.~\eqref{eq:conductance} is referred to as the \emph{flow} through $X$ and the denominator as the \emph{capacity} of $X$.
The conductance gives a measure of how hard it is to leave a small subset of vertices, minimised over the graph (where by small we mean less than half of the vertices).
Given only a graph $G$ and a target stationary distribution $\pi$, the \emph{conductance $\Phi$ of $G$ towards $\pi$} is the maximum of $\Phi(P)$ over all stochastic $P$ that satisfy the locality constraints of $G$ and whose unique stationary distribution is $\pi$.

Sinclair provided the following relationship between the conductance and mixing time~\cite[Eq. (2.13)]{Sinclair1993}
\begin{equation}\label{eq:conductance_mixing_time}
    \frac{1-2\Phi(P)}{2\Phi(P)} \log \frac{1}{\epsilon} \leq
    \mathcal{M}_\epsilon \leq
    \frac{2}{\Phi(P)^2}\qty( \log \frac{1}{\epsilon} + \log(\frac{1}{\min_i \pi_i}) ).    
\end{equation}
\begin{obs}
    The mixing time of a Markov chain is bounded between $\Omega(1/\Phi)$ and $O(1/\Phi^2)$.   
\end{obs} 

Suppose we have a lifted Markov chain $\clift{M}_G$ lifted from $M_G$, with the lifted graph $\clift{G}$ related to $G$ via the homomorphism $c$.
We define the \emph{mixing time of the marginal}, $\clift{\mathcal{M}}_{\epsilon}$, of $\clift{M}_G$, for $\epsilon >0$ as 
\begin{equation}\label{eq:marginal_mixing_time}
    \clift{\mathcal{M}}_{\epsilon} = \max_{\clift{p}{}^{(0)} \in \clift{\mathcal{P}}} \min_{T \in \mathbb{Z}_+} \qty{ T \relmiddle| \forall t \geq T,\ \norm{C \cdot (\clift{P})^t \cdot \clift{p}{}^{(0)} - \pi}_{TV} \leq \epsilon },
\end{equation}
where $\pi$ is the stationary distribution of $M_G$, $\clift{\mathcal{P}}$ is the set of allowed starting distributions of $\clift{M}_G$ and $C$ is the linear map induced by the homomorphism $c$.
Note that $\clift{\mathcal{M}}_{\epsilon} \leq \mathcal{M}_{\epsilon}$ for all $\epsilon$.
This comes from the following: we do not set $\clift{\mathcal{P}}$ as all basis states in the lifted state space, analogously to the definition of $\mathcal{M}_\epsilon$ (indeed, if this were the case we would have equality for all $\epsilon$).
Instead, we are allowed to choose a mapping from the initial state on the coarse-grained Markov chain to an initial state on the lifted chain.
The set $\clift{\mathcal{P}}$ is then the image of $\mathcal{P}$ under this mapping.
The map is chosen so as to prune the `bad' starting states from $\clift{\mathcal{P}}$ and give a faster mixing time of the marginal, yielding the inequality.
An important result of~\cite{AST17-l} is that for a lifted Markov chain to give any speedup over its coarse-grained chain we must be allowed to choose this initialisation mapping.

We shall say for a lifted chain $\clift{M_G}$ that the \emph{marginal has mixed} at a time $T$ when $\norm{C \cdot (\clift{P})^T \cdot \clift{p}{}^{(0)} - \pi}_{TV} \leq \epsilon$.
Furthermore, a lifted chain may have a marginal that has mixed without itself mixing, that is, $C (\clift{P})^t \clift{p}{}^{(0)}$ will converge to $\pi$ but $(\clift{P})^t \clift{p}{}^{(0)}$ won't necessarily converge to its stationary distribution, $\clift{\pi}$; indeed $\clift{\pi}$ doesn't even have to exist~\cite{AST17-l}.

In~\cite{Chen1999}, the authors show that for any Markov chain $M_G$, there exists a lifted chain satisfying
\begin{equation}
      \clift{\mathcal{M}}_{1/4} \geq \frac{1}{2 \Phi} 
\end{equation}
with an upper bound of $O\qty(\log(\frac{1}{\min_i \pi_i}) \frac{1}{\Phi} )$ in the case of \emph{reversible chains}, Markov chains where the flow through any cut $X\subseteq V(G)$ is the same in both directions.
In~\cite{AST17-l}, they improve the upper bound (extending to any Markov chain) to
\begin{equation} 
    \clift{\mathcal{M}}_{1/4} \leq \frac{4 \log (1/ \min_i \pi_i )}{\Phi(P)} + 2.
\end{equation}
For the bounds above, the proofs show existence of these optimal liftings, but do not provide an efficient (that is, polynomial-time) procedure to construct the lifting.
The optimal lifting in~\cite{Chen1999} relies on the solution of an $\NP$-hard problem.
Note also that these bounds hold for the case when $\mathcal{P}^c$ is taken as the set of distributions with all probability mass on any basis state in the lifted space.
We shall see later in Section~\ref{sec:d_delta_lift} that these bounds can be beaten.

Now let us consider a quantum walk on the $m$-regular graph $G$.
For a quantum walk, unitarity prevents the state itself from converging, by the following:
the $\ell_1$-norm distance between consecutive states in a quantum walk is constant, as the walk operator is unitary. Thus the limit $\lim_{t \to \infty} U^t \ket{\psi(0) }$ does not exist in general, as for convergence we demand that the distance decreases with an increasing number of timesteps.
Perhaps more naturally, we can consider the convergence of the induced probability distribution over the nodes, $Q_t(\,\cdot\, | \psi(0) )$.
We can see that this distribution does not converge either, using the following argument from~\cite{Aharonov2001}.
As the quantum walk operator $U$ is unitary, it has eigenvalues of the form $\mathe^{\mathi \theta}$.
For any $\epsilon >0$, there exists some finite $t$ for which $\abs{1 - \mathe^{\mathi \theta t}}\leq \epsilon$ for all eigenvalues $\theta$.
Thus, $U^t \ket{\psi(0) }$ can be made arbitrarily close to $\ket{\psi(0) }$ for infinitely many times $t$.
Unless $U \ket{\psi(0) } = \ket{\psi(0) }$, the walk is \emph{periodic} and $Q_t(\,\cdot\, | \psi(0) )$ does not converge. 

However, we can talk about the well-defined notion of \emph{average mixing}.
Consider the average of $Q_t(v | \psi(0) )$ over the first $T$ timesteps, $\overline{Q}_T(v | \psi(0) ) = \frac{1}{T} \sum^{T-1}_{t=0} Q_t(v | \psi(0) )$.
The limit $\lim_{T \to \infty} \overline{Q}_T(v | \psi(0) )$ exists for any $v$ and $\ket{\psi(0) }$~\cite{Aharonov2001}.
We denote by $\pi^q_{\psi(0)}$ the distribution $\lim_{T \to \infty} \overline{Q}_T(\, \cdot \, | \psi(0) )$, in analogy with classical case and refer to it as the limiting distribution of the quantum walk.
One can easily sample from the distribution $\overline{Q}_T( \,\cdot \, | \psi(0) )$ using the following procedure.
Choose a time $t\in \arange{T}$ uniformly at random, run the quantum walk for $t$ timesteps, then measure which node the walker is at.
The vertex will be distributed according to $\overline{Q}_T(\,\cdot \,| \psi(0) )$.
We will refer to $\overline{Q}_T(\,\cdot \,| \psi(0) )$ as the \emph{average mixing distribution} for the quantum walk $(U, \ket{\psi(0) })$.

We are now in a position to define the \emph{quantum mixing time},
\begin{equation}
    \mathcal{M}_{\epsilon}^q = \max_{\ket{\psi(0) } \in \varPsi } \min_{T \in \mathbb{Z}_+} \qty{ T \relmiddle| \forall t \geq T,\ \norm{\overline{Q}_T(\,\cdot \,| \psi(0) ) - \pi^q_{\psi(0)}}_{TV} \leq \epsilon },
\end{equation}
where $\varPsi$ is the allowed set of starting states, typically the basis vectors $\{ \ket{k,v} \mid k \in [m],\ v\in V(G) \}$.
Aharonov and colleagues~\cite{Aharonov2001} prove a general lower bound on the quantum mixing time $\mathcal{M}_{\epsilon}^q = \Omega(1/\Phi)$.

\section{An Example: Mixing on \boldmath{$C_n$}}\label{sec:cycle_example}

Here we take the $n$-cycle, $C_n$, to be the graph with vertex set $V(C_n) = \arange{n}$ and arc set $\{(i, i \pm 1 \operatorname{mod} n) \mid i\in\arange{n} \}$.

Consider the Markov chain $M_{C_n}$ on $C_n$, that has an arbitrary starting state in $V(C_n)$ and transition probabilities of $1/2$ on each arc.
It is well known that the mixing time of this Markov chain is quadratic in $n$ for odd $n$, i.e.\  $\mathcal{M}_{\epsilon} = \Theta(n^2 \log(1/\epsilon))$ and is undefined for even $n$.
For the cycle these transition probabilities are optimal for mixing.

We can consider a lift of this chain first considered by Diaconis, Holmes and Neal, the \emph{Diaconis lift}~\cite{Diaconis2000}.
For each vertex $i\in \arange{n}$, we augment with the pair of vertices $(\pm 1, i)$ so that $c:(\pm 1, i)\mapsto i$ and $V(\clift{C}_n) = \{ (s,k)\mid k \in \arange{n},\ s\in \pm 1 \}$.
The arc set $E(\clift{C}_n) = \{ ((s', k \pm 1 \operatorname{mod} n), (s,k)) \mid  k \in \arange{n},\ s,s'\in \pm 1  \}$.

\begin{figure}  
    \centering
    \includegraphics[width=0.65\textwidth]{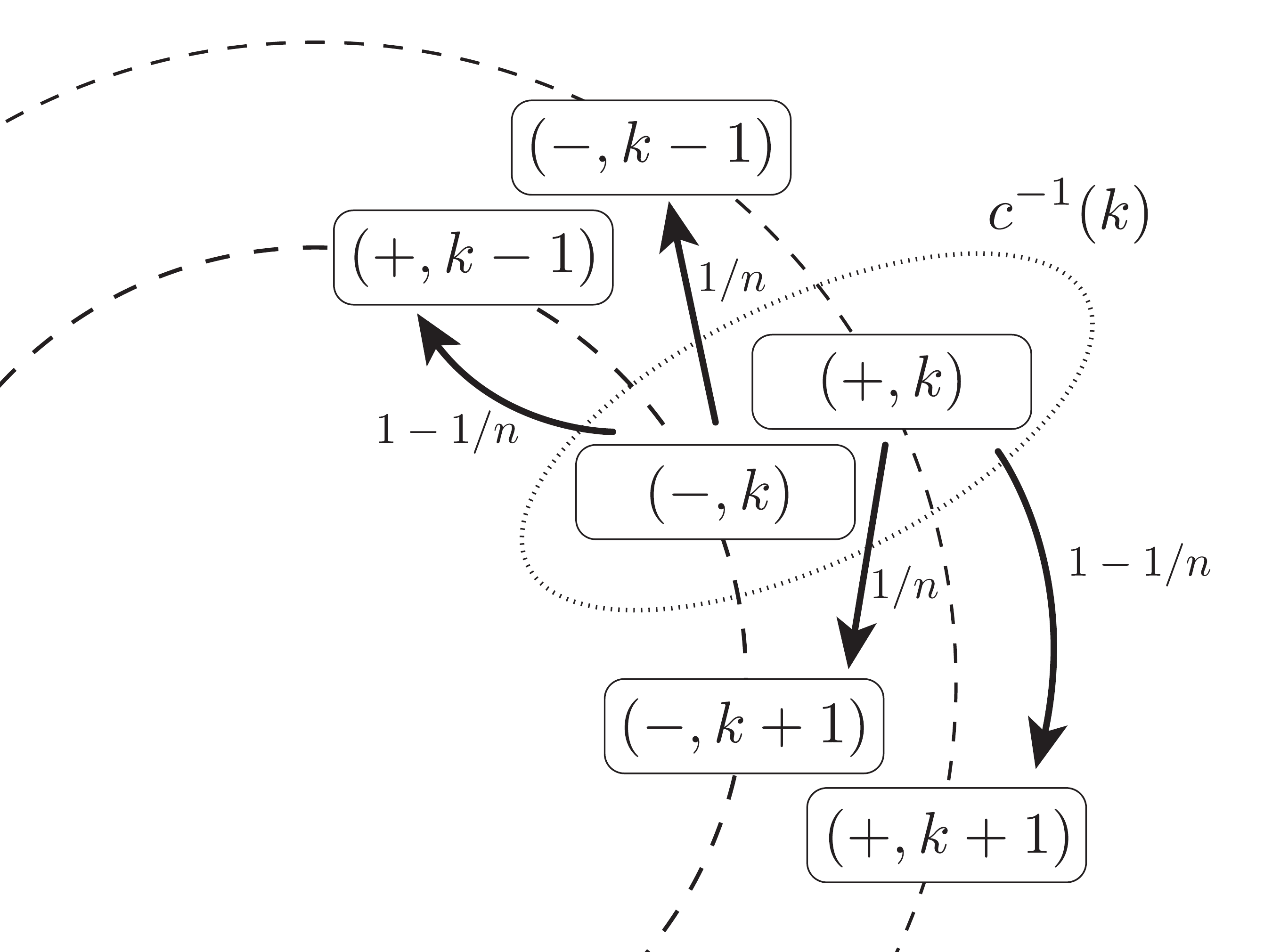}
    \caption{Illustration of the lifted Markov chain on $\clift{C}_n$.
    This figure is similar to Figure 1 in~\cite{AST17-l}.
    Each vertex $k$ is lifted to the pair of vertices $(k,-)$, $(k,+)$.
    The arrows indicate the outgoing transition probabilities from the lifted vertices of $k$, $c^{-1}(k)$.
    Note there is high probability to maintain the `sign' of the vertices as the walk progresses.}
    \label{fig:diaconislift}  
\end{figure}

The transition probabilities of the chain are as follows:
\begin{equation}
    [\clift{P}]_{i,j} = 
    \begin{cases}
        1 - 1/n, & i = (s, k),\ j = (s, k + s \operatorname{mod} n) ;\\
        1/n, & i = (s, k),\ j = (1 - s, k + s \operatorname{mod} n); \\
        0, & \text{otherwise},
    \end{cases} 
\end{equation}
where $s\in \pm 1,\ k\in \arange{n}$. 
Figure~\ref{fig:diaconislift} shows the allowed transitions and associated probabilities.
This chain has been shown to have mixing time of the marginal $\Theta(n)$ (for fixed $\epsilon$), displaying a quadratic speedup over the non-lifted chain~\cite{Diaconis2000}.
This choice of transition probabilities imposes some kind of `inertia' on the walk, in that if the walker takes a step (anti-)clockwise around the cycle, it is far more likely to take the next step (anti-)clockwise around the cycle. 

The mixing of a coined quantum walk on $C_n$ has also been studied, in~\cite{Aharonov2001}.
More concretely they perform the Hadamard walk on a Hilbert space isomorphic to $\mathbb{C}^2 \otimes \mathbb{C}^n$.
The basis states for the coin space are $\{\ket{L}, \ket{R} \}$, standing for `left' and `right'.
The coin operator is $C = H_2 = \frac{1}{\sqrt{2}}\smqty[1 & \phantom{-}1 \\ 1 & -1]$ and the shift operator $S$ acts as 
\begin{equation}
    \begin{aligned}
        S \ket{L, i} &= \ket{ L , i - 1 \operatorname{mod} n}; \\
        S \ket{R, i} &= \ket{ R , i + 1 \operatorname{mod} n}.
    \end{aligned}
\end{equation}
It is shown for this walk that the mixing time $\mathcal{M}^q_\epsilon = O(n \log(n) \frac{1}{\epsilon^3})$, demonstrating quadratic speedup in mixing as compared with the classical walk on the cycle (for fixed precision).
Interestingly, this speedup is seen in the lifted Markov chain also.
We also note that the inverse polynomial dependence on $\epsilon$ can be made inverse polylogarithmic using an amplification scheme detailed in~\cite{Aharonov2001}.

\section{An Equivalence Between Lifted Walks and Coined Quantum Walks}\label{sec:equivalence}

In this section we prove the main result of the paper and introduce the main ingredients for the proof.
In Section~\ref{sec:d_delta_lift} we introduce the lifting that will be used.
In Sections~\ref{sec:d_lift_complexity}~and~\ref{sec:pi_q_complexity} we consider the computational complexity of computing the lifting and the quantum average mixing distribution respectively.

\subsection{The \boldmath{$d$}-lifting}
\label{sec:d_delta_lift}

We now introduce another lift, which we call the \emph{$d$-lifting}, due to Apers, Ticozzi and Sarlette~\cite{AST17-l}. 
We shall omit full details of how the lift is constructed and the homomorphism $d $ for brevity. 
First, we need the following definitions. 
Let $G$ be a connected, directed graph.
The \emph{distance}, $d(u,v)$, between the nodes $u$ and $v$ in $G$ is the shortest length path between them.
The \emph{diameter} of $G$, $D(G)$, is the greatest distance between any pair of vertices in $G$, or rather
\begin{equation}
    D(G):= \max_{u,v \in V(G)} d(u,v).    
\end{equation}
The \emph{tensor product} of graphs $G$ and $H$, denoted by $G\otimes H$, has vertex set $V(G)\times V(H)$ and an arc $( (i, j),(k, l))$ if and only if$( i, k)\in E(G)$ and $(j, l)\in E(H)$.

\begin{prop}\label{prop:ATS_lift}
    \emph{($d $-lifting~\cite[Theorem 2]{AST17-l})}
    Let $M_G=(P,p^{(0)})$ be a Markov chain on a connected graph $G$ on $n$ vertices.
    Moreover, $P$ has stationary distribution $\pi$ with all strictly positive elements.
    Then, there exists an $d $-lifted Markov chain, $\dlift{M_G}$, on a graph $\dlift{G}$ having $D(G)\cdot n^2$ vertices, for which $\dlift{\mathcal{M}}_{\epsilon} \leq D(G)$, where $D(G)$ is the diameter of $G$ and $\epsilon > 0$ is arbitrary. 
\end{prop}
\begin{obs}
    The lifted chain's marginal mixes to $\pi$ in $D(G)$ timesteps, to arbitrary precision $\epsilon$, a remarkable fact.
\end{obs}

In their statement of this theorem in~\cite{AST17-l}, the authors stipulate that this lift has certain restrictive properties.
The first is that the starting distribution for the lift is initialised according to a particular mapping $d^{\text{init}}: p^{(0)} \mapsto \dlift{p}{}^{(0)}$ i.e.\ $\mathcal{P}^{d} = \qty{d^{\text{init}} p^{(0)} \relmiddle| p^{(0)} \in \mathcal{P} }$ in the definition of mixing of the marginal~\eqref{eq:marginal_mixing_time}.
The proposition does not contradict the conductance lower bound given in~\cite{Chen1999} as discussed earlier, because this is defined for $\mathcal{P}^c$ being the set of \emph{all} probability distributions over the lifted vertices (or equivalently, distributions with all probability concentrated at basis states).
Indeed, our definition of a lifted chain allows us this choice, since $\dlift{p}{}^{(0)}$ satisfies $p^{(0)} = d  \dlift{p}{}^{(0)}$ by construction, where $d $ is the linear map induced by the lift homomorphism..
The second restriction is that the lifted chain having a marginal that has mixed does not necessarily imply that the lifted chain itself has mixed.
Since we will only care about the marginal mixing to $\pi$, this is not important for us, in fact the Diaconis lift on the cycle discussed in Section~\ref{sec:cycle_example} has this property.
We must also allow for the lifted chain to be reducible, that is, $\dlift{G}$ is not a connected graph.
Again, this does not concern us. 

For completeness, we shall briefly describe the $d$-lifting, without going into exhaustive detail.
The interested reader is referred to~\cite{AST17-l}.
The lifting rests on the following construction from~\cite{AST17-l}: let $G$ be a graph on $n$ vertices and let $p, p'$ be probability distributions on $V(G)$.
Then, there is a set of $D(G)$ transition matrices on $G$, $\{ P(i) \}^{D(G)}_{i=1}$, called a \emph{stochastic bridge} such that 
$p' = P(D(G)) P(D(G)-1) \cdots P(2) P(1) p$.

To apply the $d$-lifting, for each vertex of $G$ we create a copy of the graph $G \otimes P_T$, where $T= D(G)$, then take their disjoint union, giving the graph $\dlift{G} := \biguplus_{v_0 \in V(G)} G \otimes P_T$.
The vertex set $V(\dlift{G}) = \{ (t, v_0, v)\mid t \in \{0,1,\ldots, D(G) -1\},\ v_0,v \in V(G) \}$ and $d : (t,v_0,v) \mapsto v$.
We provide an example diagram of the lift in Figure~\ref{fig:d_delta_lift}.

\begin{figure}[hbtp]
        \centering
        \includegraphics[width=\textwidth]{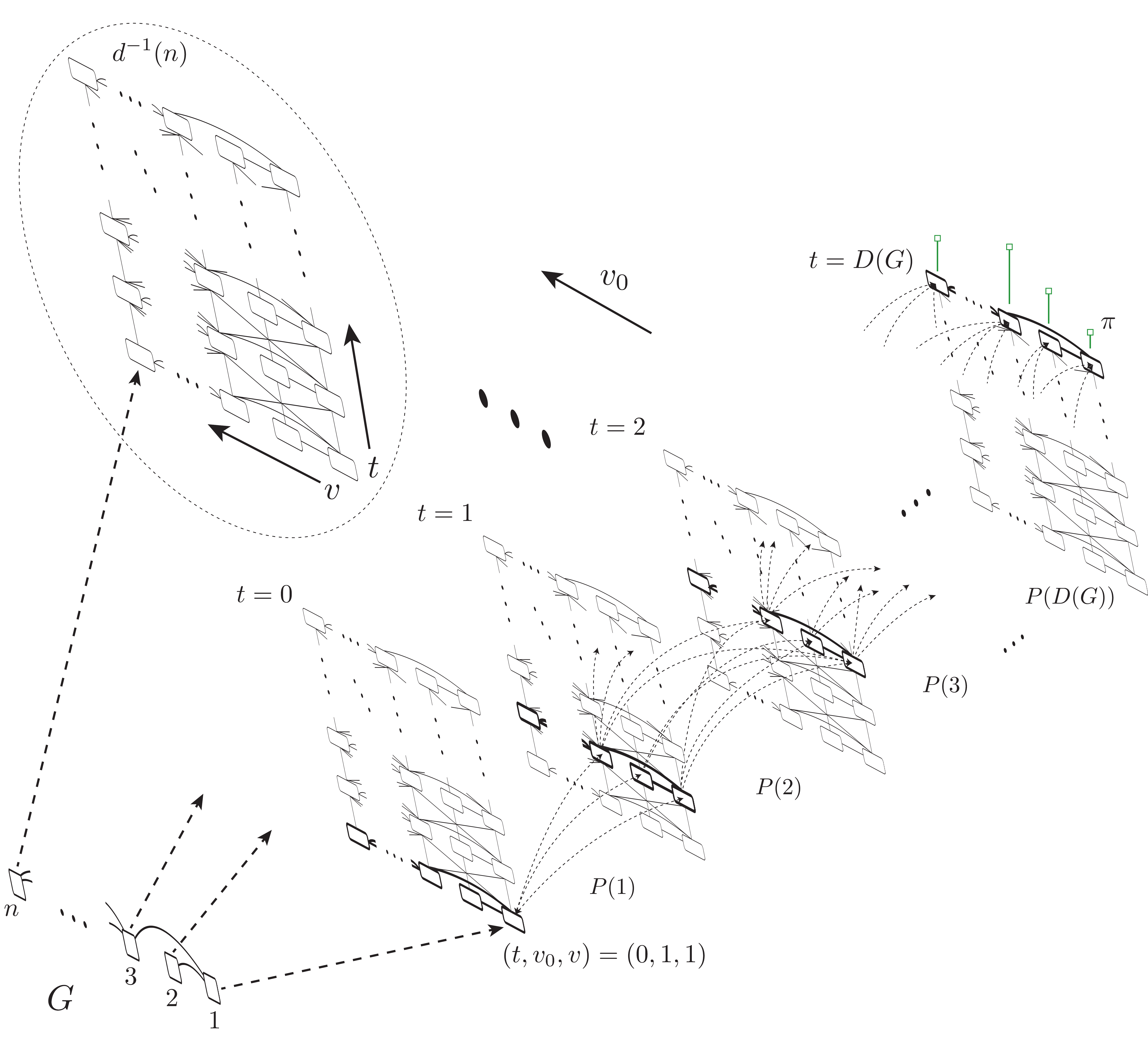}
        \caption{Illustration of the $d$-lifting.
        The figure is the same as Figure 4 in~\cite{AST17-l}, with some more detail.
        For each vertex in $V(G)$ there is a single disjoint copy of the graph $G \otimes P_{D(G)}$.
        At the top left we show $d^{-1}(n)$ for $n \in V(G)$.
        On the bottom right we illustrate the time evolution of a walk starting at vertex $1 \in V(G)$, whose corresponding starting state in the $d$-lifted walk is $(t,v_0,v) = (0,1,1)$.
        The evolution, defined by a stochastic bridge $\{ P(i) \}^{D(G)}_{i=1}$ is depicted with multiple copies of $G \otimes P_{D(G)}$, one for each $t$, with the corresponding $v$ vertices boldened.
        We see that the final distribution is $\pi$ over the appropriate vertices, indicated by the green lines capped with boxes.
        After $t \geq D(G)$ timesteps, the marginal of the lifted chain has mixed to $\pi$ and the dynamics proceed according the coarse-grained transition matrix $P$.
        We also indicate with arrows labelled $t$, $v_0$ and $v$ the direction in which the lifted vertex indices $(t,v_0,v)$ increase.
        Self-loops in the drawing are omitted for clarity.
        }
        \label{fig:d_delta_lift} 
    \end{figure}

Roughly speaking the lifted walk starts by sampling a vertex, $X(0)$, from $G$ according to the initial distribution $p^{(0)}$.
Then, we walk on the $X(0)$\supth copy of $G \otimes P_{D(G)}$, starting at the node $(0,  X(0), X(0))$.
The transition probabilities are engineered using the stochastic bridge such that $t$ increases by one at each timestep and $P(t)$ is applied to the $v$ space at timestep $t$.
This ensures that the final distribution in the $v$ space is $\pi$, the stationary distribution of the chain $M_G$, by taking $p' = \pi$ and $p = \vb{e}_{X(0)}$ in Eq.~\eqref{eq:d_g_bridge}.
Marginalising gives us \emph{exactly} the marginal distribution $\pi$ after $D(G)$ timesteps.
In practise, the stochastic bridge will be attained to some arbitrary precision $\delta$, so we have that the marginal mixes to $\pi$ for arbitrary $\delta$.

\subsection{Complexity of computing the \boldmath{$d $}-lifting}\label{sec:d_lift_complexity}

Let us revisit the following claim, made in \cite{AST17-l} and proved in the newer paper~\cite{AST17-q}.

\begin{claim}[Stochastic Bridge]\label{cla:stochastic_bridge}
    Let $G$ be a connected graph on $n$ vertices and let $p, p'$ be probability distributions on $V(G)$.
    Then, there is a set of $D(G)$ transition matrices on $G$, $\{ P(i) \}^{D(G)}_{i=1}$, called a \emph{stochastic bridge} such that 
    \begin{equation}\label{eq:d_g_bridge}
        p' = P(D(G)) P(D(G)-1) \cdots P(2) P(1) p,
    \end{equation} 
    where $D(G)$ is the diameter of $G$.   
\end{claim} 

Claim~\ref{cla:stochastic_bridge} is proved in~\cite[Lemma 5]{AST17-q}, taking inspiration from Aaronson~\cite{Aaronson2005}. Their proof involves showing that the transition probability matrix $P(t)$ between `times' $t$ and $t+1$ are given by the solution of a maximum flow problem.
One then solves this maximum flow problem for each $t \in (0,1,\ldots, D(G)-1)$ to obtain the transition probabilities $P(1), P(2), \ldots P(D(G))$, see Figure~\ref{fig:atslemma5}. 

This proof is lacking one ingredient to be completely constructive, a `schedule' of flows for the edges adjacent to the source and sink vertices at a given pair of times $(t,t+1)$. To use the notation of~\cite[Lemma 5]{AST17-q}, we need to set a schedule of $p_t$, i.e. $\{ p_0, p_1, \ldots, p_{D(G)} \}$ 
for a given vertex $i$.

\begin{figure}[hbtp!]
    \centering
    \includegraphics[width=0.85\textwidth]{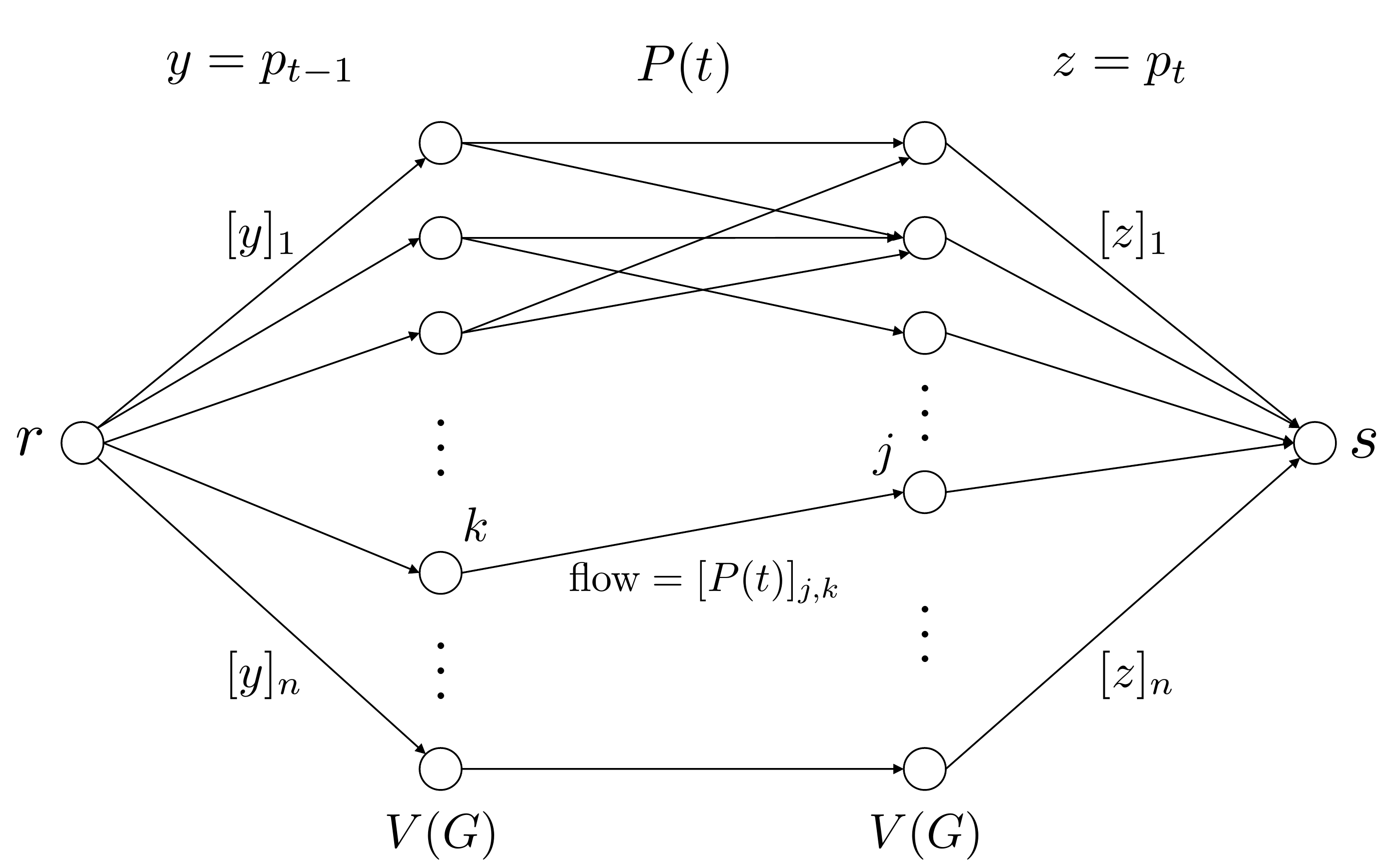}
    \caption{Illustration of the construction used for proving the existence of the stochastic bridge in~\cite[Lemma 5]{AST17-q}, see their paper full details.
    We wish to solve the $(r,s)$-flow shown in the figure, yielding $P(t)$ for each $t$ in the stochastic bridge, that is, we need to solve a maximum flow for each of the $D(G)$ transition matrices in the stochastic bridge.
    The solution flows on the arcs in between the two copies of $V(G)$ give the transition matrix $P(t)$. We set the capacities $y$ and $z$, the capacities through the middle arcs are all $1$.
    }
    \label{fig:atslemma5}
\end{figure}

For a given vertex $i$, we can find the stochastic bridge taking $p_0 = e_i$ to $p_{D(G)} = \pi$ as follows: find a spanning tree $T_i$ of $G$ rooted at vertex $i$, using breadth-first search. We shall now modify $T_i$ using the following procedure: walk through $T_i$ using breadth-first search. Every time you reach any vertex $j$ that has already been visited, append a new child vertex also labelled $j$. We call this modified graph $T'_{i}$ ; it is still a tree. Moreover, we denote the $t^{\text{th}}$ \emph{level} of $T_i'$, $\ell_t(T_i')$, the set of vertices (in $V(G)$) at distance $t$ from the root $i$ in $T_i'$.

As an example, consider the graph $G$ in Figure~\ref{fig:spanning_trees}, with the spanning tree $T_1$ (rooted at vertex 1) and the related tree $T'_{1}$.
Also, let $\mathcal{D}(j, T)$ be the set of descendent leaves of the vertex $j$ in a tree $T$. 
We then set our `bridge schedule' as follows:
\begin{equation}\label{eq:sched_prob}
[p_t]_j =
\begin{cases} 
\sum_{k \in \mathcal{D}(j, T_i')}[\pi]_{k}, &j \in \ell_t(T_i');\\
0, &  \text{otherwise}.
\end{cases}
\end{equation}

\begin{figure}[hbtp!]
    \centering
    \includegraphics[width=0.75\textwidth]{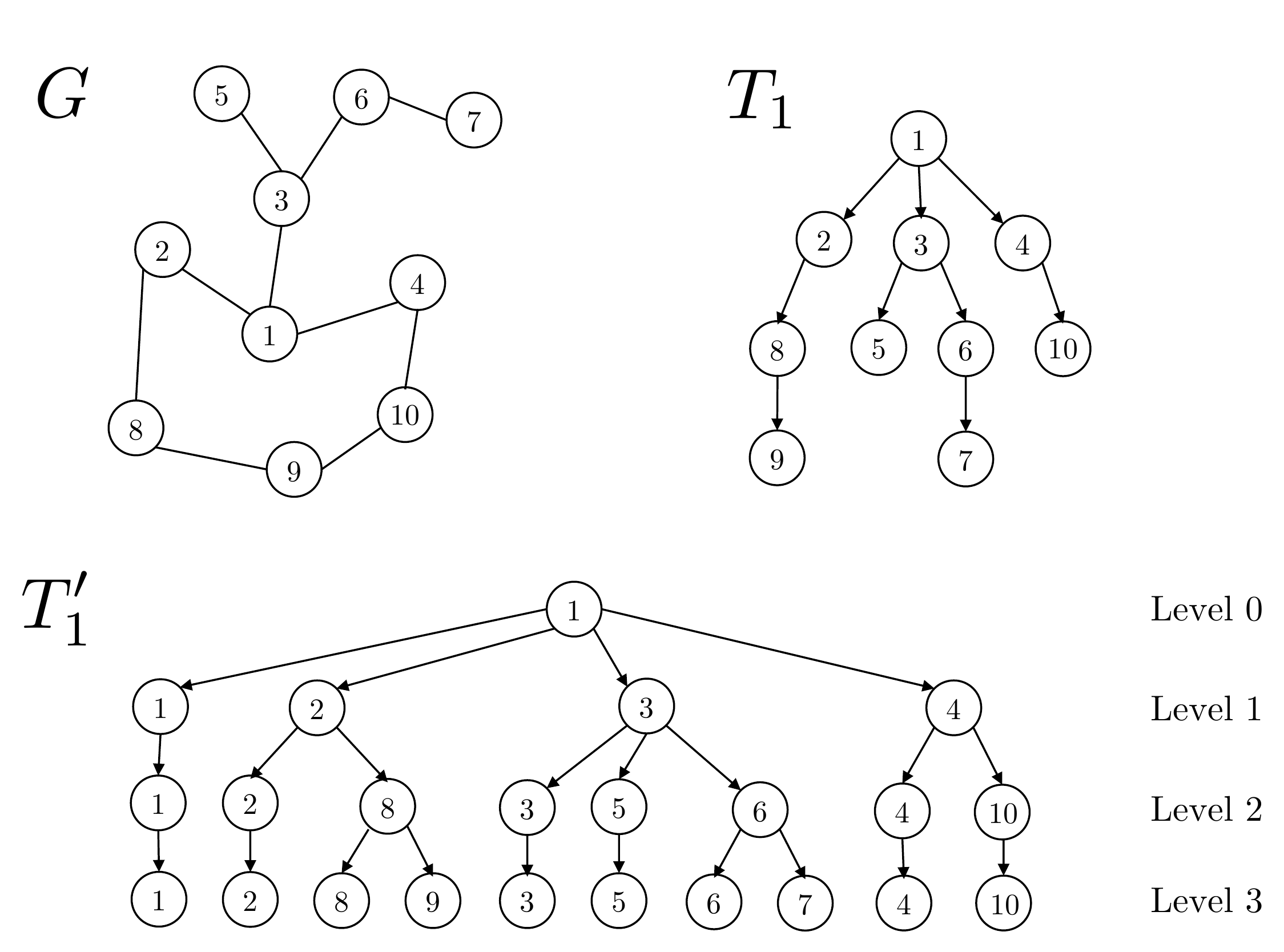}
    \caption{Illustration of the construction of the graph $T'_i$ for a graph $G$.
    First, find a spanning tree $T_i$ of $G$ rooted at vertex $i$, using breadth-first search. Then, modify $T_i$ using the following procedure: walk through $T_i$ using breadth-first search. Every time you reach any vertex $j$ that has already been visited, append a new child vertex also labelled $j$.
    The resulting graph is $T'_i$.}
    \label{fig:spanning_trees}
\end{figure}

This schedule effectively routes probability mass through the lifted graph at each timestep. We then solve each max-flow problem with (\cite[Lemma 5]{AST17-q} notation) $y = p_t$, $z = p_{t+1}$ for each $t \in \{0,1,\ldots, D(G)-1\}$ to get the transition probabilities.
Notice here, it is possible to `prune' the vertices in the $d$-lifted state space that do not occur at the $t^{\text{th}}$ level of each $T_i'$, that is, vertices for which $[p_t]_j = 0$. To avoid complications we shall not take this into account in the analysis that follows.

Having assigned the `schedule probabilities' for the stochastic bridge construction, we can prove the following.
\begin{lem}\label{lem:d_lift_runtime}
    Let $M_G=(P,p^{(0)})$ be a Markov chain on a connected graph $G$ on $n$ vertices.
    Moreover, $P$ has stationary distribution $\pi$ with all strictly positive elements.
    Then, computing the transition probabilities of the $d$-lifted Markov chain, $M^{d}_G$ requires $O(n^4 D(G))$ time, where $D(G)$ is the diameter of $G$. 
\end{lem}
\begin{proof}
    For the $d$-lifting of a Markov chain on an $n$ vertex graph, we are required to compute a stochastic bridge corresponding to each vertex, with $p' = \pi$ for every bridge and $p = \vb{e}_i$ for the $i$\supth vertex. 
    We require $n$ stochastic bridges, each containing $D(G)$ $n\times n$ transition matrices.
    
    Solving for each transition matrix requires solving a max-flow problem on $2n + 2$ vertices, where certain flows are given by the `schedule' probabilities~\eqref{eq:sched_prob}.
    Computing the schedule probabilities $[p_t]_j$ for a given vertex $i$ involves finding a spanning tree $T_i$, walking through $T_i$, appending vertices, then for each $t\in [D(G) -1]$ summing up the values of the children. The complexity of this task is $O(n D(G))$.
    Taken over all $n$ stochastic bridges, computing the schedule probabilities takes time $O(n^2 D(G))$.
     
    Maximum-flow problems can be solved in time $O(|V(G)|^3)$ (see Malhotra, Pramodh Kumar, Maheshwari~\cite{Malhotra1978}), so the total runtime complexity for solving the max flow problems is $O(n^4 D(G))$ as we solve $nD(G)$ max-flow problems on graphs with $2n + 2$ vertices.
    This is the overall runtime complexity of computing the transition probabilities for the $d$-lifting on a graph $G$ with $n$ vertices, as solving the max-flow problems washes out the complexity of computing the schedule probabilities.
\end{proof}

\subsection{Computing the quantum average mixing distribution}\label{sec:pi_q_complexity}

First, we quote a useful identity from~\cite[Theorem 6.1]{Godsil2017} concerning the elements of the quantum average mixing distribution of a quantum walk on a $m$-regular graph $G$ on $n$ vertices:

\begin{equation}\label{eq:pi_q_av_mix_mat}
    [\pi^q_{\psi(0)}]_v = \sum_r \langle \psi(0) | F_r D_v F_r |\psi(0)\rangle ,
\end{equation}
where $D_v\in \mathbb{R}^{mn \times mn}$ is the diagonal matrix with a $1$ in positions corresponding to vertex $v$ and zeros elsewhere. and $F_r$ are the idempotents of the spectral decomposition of $U$. 
Thus, knowing the spectral decomposition of $U$ allows us to compute the quantum average mixing distribution of the walk, $\pi^q_{\psi(0)}$.

Let us now consider the computational complexity of computing the quantum average mixing distribution, $\pi^q_{\psi(0)}$, using Eq.~\eqref{eq:pi_q_av_mix_mat}.
Computing the spectral decomposition of the average mixing matrix takes time $O((nm)^3)$. Each term $F_r^\dagger D_k F_r$ is $O((nm)^2 \cdot m)$ since $D_k$ has only $m$ non-zero terms. Taking $\langle \psi(0) | \cdots |\psi(0)\rangle$ is an additive $O(m)$ factor. We then have that $r$ can range up to $nm$ and we perform the sum for $v \in [n]$ for a total $O((nm)^4)$. Now, taking $m= O(n)$ we have that computing $\pi^q(\psi(0))$ requires runtime $O(n^8)$.
This gives us the following lemma.

\begin{lem}\label{lem:pi_q_runtime}
    Let $(U, |\psi(0) \rangle )$ be a coined quantum walk on a $m$-regular graph $G$.
    Then, computing the quantum average mixing distribution, $\pi^q_{\psi(0)}$, takes time $O(n^4 m^4) = O(n^8)$.
\end{lem}

We will also need the following result concerning the quantum average mixing distribution.

\begin{lem}\label{lem:pi_q_nonzero}
    Let $(U, |\psi(0) \rangle )$ be a coined quantum walk on a connected $m$-regular graph $G$.
    Then, every element of the quantum average mixing distribution $\pi^q_{\psi(0)}$ is strictly positive.
\end{lem}
\begin{proof}
    \makebox[0.915\linewidth][s]{Recall Eq.~\eqref{eq:pi_q_av_mix_mat}, that states $[\pi^q_{\psi(0)}]_v = \sum_r \langle \psi(0) | F_r D_v F_r |\psi(0)\rangle $, where $D_v = $}
    \newline$\sum_{k \in [m]} \ketbra{k, v}$ and $F_r$ are the idempotents of the spectral decomposition of $U$.
    Notice that 
    \begin{equation}
        \langle \psi(0) | F_r D_v F_r |\psi(0)\rangle = \sum_{k \in [m]} \bra{\psi(0)} F_r \ketbra{k, v} F_r \ket{\psi(0)} = \sum_{k \in [m]} \abs{\bra{\psi(0)} F_r \ket{k, v}}^2.
    \end{equation}
    Thus, if $[\pi^q_{\psi(0)}]_v = 0$, then $\sum_r \sum_{k \in [m]} \abs{\bra{\psi(0)} F_r \ket{k, v}}^2 = 0$ and there exist $u \in V(G)$, $j,k \in [d]$ such that $\sum_r \abs{\bra{j, u} F_r \ket{k, v}}^2 = 0$.
    This implies that $[F_r]_{(j, u), (k,v)} = 0$ for all $r$ and any linear combination of the $E_r$ has a $((j, u), (k,v))$-component of zero.
    In this case, for every $t \in \mathbb{N}$ we have $\bra{j,u} U^t \ket{k,v} = 0$.
    Now this is only true if $G$ is not connected, as it implies there is no path in $G$ of the form $(f_v(k), \ldots, w)$, where $w$ is the vertex satisfying $f_w(j) = u$.
    By contraposition we infer that $G$ being connected implies that $[\pi^q_{\psi(0)}]_v >0$ for all $v\in V(G)$. 
\end{proof}

This result should not be surprising since the limiting distribution of a classical ergodic Markov chain has strictly positive elements.

\subsection{Main Result}\label{sec:main_result}

We now have all of the pieces to prove the main result.

\begin{thm}\label{thm:main_result}
    Let $(U, |\psi \rangle )$ be a coined quantum walk on a connected $m$-regular graph $G$ on $n$ vertices.
    Then there exists a lifted Markov chain on $n^2 D(G)$ vertices with marginal that mixes exactly to the quantum average mixing distribution $\pi^q_{\psi(0)}$ after $D(G)$ timesteps, where $D(G)$ is the diameter of $G$. Computing the transition probabilities for the lifted Markov chain requires $O(n^4(m^4 + D(G)))$ time.
\end{thm}
\begin{proof}
    We will use the $d$-lifting of the Markov chain on $G$ with $\pi^q_{\psi(0)}$ as the target distribution.
    From Lemma~\ref{lem:pi_q_nonzero}, $\pi^q_{\psi(0)}$ has strictly positive elements, and so satisfies the conditions for the $d$-lifting of Proposition~\ref{prop:ATS_lift}.
    From Lemmas~\ref{lem:d_lift_runtime}~and~\ref{lem:pi_q_runtime} we have that the runtime of computing $\pi^q_{\psi(0)}$ is $O(n^4 m^4)$ and the computing the $d$-lifting takes $O(n^4 D(G))$ time.
\end{proof}

Indeed, we can also generalise this result to a general quantum walk in the following way.

\begin{cor}\label{cor:main_result}
    Let $(U, |\psi \rangle )$ be a general quantum walk on a connected graph $G$.
    Then there exists a lifted Markov chain on $n^2 D(G)$ vertices with marginal that mixes exactly to the quantum average mixing distribution $\pi^q_{\psi(0)}$ after $D(G)$ timesteps, where $D(G)$ is the diameter of $G$. Computing the transition probabilities for the lifted Markov chain requires $O(n^8)$ time.
\end{cor}

Here, we take $m = O(n)$, $D(G) = O(n)$ and notice that the proofs of Lemmas~\ref{lem:pi_q_nonzero}~and~\ref{lem:pi_q_runtime} hold for general quantum walks also.

\section{Discussion and Open Questions} \label{sec:discussion}

We have demonstrated that if one wants to use a quantum walk for its mixing properties, i.e.\  use the quantum walk to sample from a given probability distribution, then there always exists a lifted Markov chain that mixes in time upper bounded by the number of vertices in the graph.
Moreover, we can compute the transition probabilities for the lifting in polynomial time.
The lifted Markov chain takes place on a state space that is polynomially larger than in the quantum case.
In some sense this gives us an upper bound on the amount of classical resources that are needed to simulate a quantum walk with a classical random process.

This work suggests a number of open questions for further research.
Some key questions to be answered are
\begin{itemize}
    \item Is the $d$-lifting optimal in terms of the number of states and computational complexity for achieving diameter-time mixing?
    \item What resources do we need to give a quantum walk to achieve diameter-time mixing in general? Is this possible? 
    Perhaps fewer resources than the $d$-lifting uses are required.      
\end{itemize}    

The first question could perhaps be approached in the first case by pruning vertices in the $d$-lifting that are unnecessary.
The second question would require engineering some kind of `quantum lifting', after suitably defining such a lifting.
In this case it could be possible to see diameter-time mixing with fewer computational resources consumed than in the classical case.
On the other hand this might be impossible, which would be more interesting still.

\section*{Acknowledgements}

The author thanks Simone Severini, their PhD adviser, for suggesting the direction of this work and many useful comments on the manuscript, also providing the references~\cite{AST17-a, Richter2007}.
M$\bar{\mbox{a}}$ris Ozols has provided extensive commentary on the manuscript, alongside anonymous reviewers, for which the author is very grateful.
The author also thanks Alberto Ottolenghi, Leonard Wossnig, Andrea Rochetto and Joshua Lockhart for useful discussions.
This work was supported by the EPSRC Centre for Doctoral Training in Quantum Technologies.

\bibliographystyle{utils/nicebib-alpha}
\bibliography{lifting_and_quantum_walks.bib,extra_refs.bib}

\begin{thebibliography}{AdFDJ03}
\providecommand{\url}[1]{\texttt{#1}}
\providecommand{\urlprefix}{}
\expandafter\ifx\csname urlstyle\endcsname\relax
  \providecommand{\doi}[1]{doi:\discretionary{}{}{}#1}\else
  \providecommand{\doi}{doi:\discretionary{}{}{}\begingroup
  \urlstyle{rm}\Url}\fi
\providecommand{\arxivId}[2][]{\texttt{#2}}

\bibitem[AAKV01]{Aharonov2001}
D.~Aharonov, A.~Ambainis, J.~Kempe and U.~Vazirani.
\newblock \href{http://dx.doi.org/10.1145/380752.380758}{\emph{{Quantum walks
  on graphs}}}.
\newblock \href{http://portal.acm.org/citation.cfm?doid=380752.380758}{In
  \href{http://portal.acm.org/citation.cfm?doid=380752.380758}{\textsl{Proc.
  thirty-third Annu. ACM Symp. Theory Comput. - STOC '01}}},
  \textsl{\href{http://portal.acm.org/citation.cfm?doid=380752.380758}{50--59}}.
  ACM Press, New York, NY, USA (2001).

\bibitem[Aar05]{Aaronson2005}
S.~Aaronson.
\newblock \href{http://dx.doi.org/10.1103/PhysRevA.71.032325}{\emph{{Quantum
  computing and hidden variables}}}.
\newblock \href{http://dx.doi.org/10.1103/PhysRevA.71.032325}{\textsl{Phys.
  Rev. A}},
  \textsl{\href{http://dx.doi.org/10.1103/PhysRevA.71.032325}{\textbf{71},
  032325}} (2005).

\bibitem[AdFDJ03]{Andrieu2003}
C.~Andrieu, N.~de~Freitas, A.~Doucet and M.~I. Jordan.
\newblock \href{http://dx.doi.org/10.1023/A:1020281327116}{\emph{{An
  Introduction to MCMC for Machine Learning}}}.
\newblock \href{http://dx.doi.org/10.1023/A:1020281327116}{\textsl{Mach.
  Learn.}},
  \textsl{\href{http://dx.doi.org/10.1023/A:1020281327116}{\textbf{50}, 5}}
  (2003).

\bibitem[AST17-a]{AST17-a}
S.~Apers, A.~Sarlette and F.~Ticozzi.
\newblock \href{https://hal.inria.fr/hal-01395592/document}{\emph{Fast Mixing
  with Quantum Walks vs. Classical Processes}}.
\newblock \href{https://hal.inria.fr/hal-01395592/document}{In
  \href{https://hal.inria.fr/hal-01395592/document}{\textsl{Quantum Information
  Processing (QIP) 2017}}}. Seattle, United States. \texttt{<hal-01395592>}
  (2017).

\bibitem[AST17-q]{AST17-q}
S.~Apers, A.~Sarlette and F.~Ticozzi.
\newblock \href{https://arxiv.org/pdf/1712.01609.pdf}{\emph{{Simulation of
  Quantum Walks and Fast Mixing with Classical Processes}}} (2017).
\newblock
  \href{https://arxiv.org/pdf/1712.01609.pdf}{\arxivId{arXiv:1712.01609}}.

\bibitem[ATS17-l]{AST17-l}
S.~Apers, F.~Ticozzi and A.~Sarlette.
\newblock \href{http://arxiv.org/abs/1705.08253}{\emph{{Lifting Markov Chains
  To Mix Faster: Limits and Opportunities}}} (2017).
\newblock \href{http://arxiv.org/abs/1705.08253}{\arxivId{arXiv:1705.08253}}.

\bibitem[CLP99]{Chen1999}
F.~Chen, L.~Lov{\'{a}}sz and I.~Pak.
\newblock \href{http://dx.doi.org/10.1145/301250.301315}{\emph{{Lifting Markov
  chains to speed up mixing}}}.
\newblock \href{http://portal.acm.org/citation.cfm?doid=301250.301315}{In
  \href{http://portal.acm.org/citation.cfm?doid=301250.301315}{\textsl{Proc.
  thirty-first Annu. ACM Symp. Theory Comput. - STOC '99}}},
  \textsl{\href{http://portal.acm.org/citation.cfm?doid=301250.301315}{275--281}}.
  ACM Press, New York, New York, USA (1999).

\bibitem[DHN00]{Diaconis2000}
P.~Diaconis, S.~Holmes and R.~M. Neal.
\newblock \href{http://dx.doi.org/10.1214/aoap/1019487508}{\emph{{Analysis of a
  nonreversible Markov chain sampler}}}.
\newblock \href{http://dx.doi.org/10.1214/aoap/1019487508}{\textsl{Ann. Appl.
  Probab.}},
  \textsl{\href{http://dx.doi.org/10.1214/aoap/1019487508}{\textbf{10}, 726}}
  (2000).

\bibitem[GZ17]{Godsil2017}
C.~Godsil and H.~Zhan.
\newblock \href{https://arxiv.org/pdf/1701.04474.pdf}{\emph{{Discrete-Time
  Quantum Walks and Graph Structures}}} (2017).
\newblock
  \href{https://arxiv.org/pdf/1701.04474.pdf}{\arxivId{arXiv:1701.04474}}.

\bibitem[KF09]{Koller2009}
D.~Koller and N.~Friedman.
\newblock \href{https://mitpress.mit.edu/books/
  probabilistic-graphical-models}{\emph{{Probabilistic graphical models :
  principles and techniques}}}.
\newblock MIT Press (2009).

\bibitem[LPW09]{Levin2009}
D.~A. Levin, Y.~Y. Peres and E.~L. Wilmer.
\newblock \href{http://bookstore.ams.org/mbk-58/}{\emph{{Markov chains and
  mixing times}}}.
\newblock American Mathematical Society (2009).

\bibitem[MK82]{Molloy1982}
Molloy and M.~K.
\newblock \href{http://dx.doi.org/10.1109/TC.1982.1676110}{\emph{{Performance
  Analysis Using Stochastic Petri Nets}}}.
\newblock \href{http://dx.doi.org/10.1109/TC.1982.1676110}{\textsl{IEEE Trans.
  Comput.}},
  \textsl{\href{http://dx.doi.org/10.1109/TC.1982.1676110}{\textbf{C-31}, 913}}
  (1982).

\bibitem[MKM78]{Malhotra1978}
V.~Malhotra, M.~Kumar and S.~Maheshwari.
\newblock \href{http://dx.doi.org/10.1016/0020-0190(78)90016-9}{\emph{{An
  $O(|V|^3)$ algorithm for finding maximum flows in networks}}}.
\newblock \href{http://dx.doi.org/10.1016/0020-0190(78)90016-9}{\textsl{Inf.
  Process. Lett.}},
  \textsl{\href{http://dx.doi.org/10.1016/0020-0190(78)90016-9}{\textbf{7},
  277}} (1978).

\bibitem[Ric07]{Richter2007}
P.~C. Richter.
\newblock \href{http://dx.doi.org/10.1103/PhysRevA.76.042306}{\emph{{Quantum
  speedup of classical mixing processes}}}.
\newblock \href{http://dx.doi.org/10.1103/PhysRevA.76.042306}{\textsl{Phys.
  Rev. A}},
  \textsl{\href{http://dx.doi.org/10.1103/PhysRevA.76.042306}{\textbf{76},
  042306}} (2007).

\bibitem[Sin93]{Sinclair1993}
A.~Sinclair.
\newblock
  \href{http://link.springer.com/10.1007/978-1-4612-0323-0}{\emph{{Algorithms
  for Random Generation and Counting: A Markov Chain Approach}}}.
\newblock Birkh{\"{a}}user Boston, Boston, MA (1993).

\end{thebibliography}
\end{document}